\documentclass{amsart}
\usepackage{bbm}
\usepackage{cite}

 \usepackage{amssymb}
 \usepackage{amsmath,amsthm}
 \usepackage{latexsym}
 \usepackage{amsfonts}
 \usepackage{graphicx}

\newtheorem{theorem}{Theorem}[section]
\newtheorem{lemma}[theorem]{Lemma}
\newtheorem{corollary}[theorem]{Corollary}
\newtheorem{proposition}[theorem]{Proposition}

\theoremstyle{definition}
\newtheorem{definition}[theorem]{Definition}
\newtheorem{example}[theorem]{Example}

\theoremstyle{remark}
\newtheorem{remark}[theorem]{Remark}

\renewcommand{\a}{\alpha}
\renewcommand{\b}{\beta}

\newcommand{\s}{\sigma}

\newcommand{\M}{{\cal M}}

\newcommand{\ty}{\infty}

\newcommand{\f}{\varphi}
\newcommand{\ov}[1]{\overline{#1}}

\renewcommand{\O}{\Omega}
\newcommand{\pa}{\partial}

\newcommand{\eR}{\mathbb{R}}

\newcommand{\Ze}{\mathbb{Z}}

\newcommand{\Ce}{\mathbb{C}}

\newcommand{\re}{\mathop{\mathrm{Re}}}

\newcommand{\po}{{\mathop{\mathcal P}}}

\newcommand{\res}{\operatorname{res}}

\newcommand{\D}{\mathrm{d}}
\newcommand{\I}{\mathbbm{i}}

\newcommand{\ovb}[1]{\mkern 1.5mu\overline{\mkern-1.5mu#1\mkern-1.5mu}\mkern 1.5mu}
\newcommand{\unb}[1]{\mkern 1.5mu\underline{\mkern-1.5mu#1\mkern-1.5mu}\mkern 1.5mu}

\newcommand{\di}{\,\mathrm{d}}

\newcommand{\cal}{\mathcal}

\begin{document}


\title[Fractality and Lapidus zeta functions at infinity]{Fractality and Lapidus zeta functions at infinity\thanks{The research of Goran Radunovi\'c was supported in part by the Croatian Science Foundation under the project IP-2014-09-2285 and by the Franco-Croatian PHC-COGITO project.}}


\author[Goran Radunovi\'c]{Goran Radunovi\'c}

%
%
%

\begin{abstract}
We study fractality of unbounded sets of finite Lebesgue measure at infinity by introducing the notions of Minkowski dimension and content at infinity. We also introduce the Lapidus zeta function at infinity, study its properties and demonstrate its use in analysis of fractal properties of unbounded sets at infinity.
\end{abstract}



\maketitle


\section{Introduction}

In this paper we are interested in relative fractal drums $(A,\O)$ in which the set $A$ has degenerated to the point at infinity.
In short, a relative fractal drum $(A,\O)$ generalizes the notion of a bounded subset of $\eR^N$ and is defined as an ordered pair of subsets of $\eR^N$ where $A$ is nonempty and $\O$ is of finite $N$-dimensional Lebesgue measure satisfying a mild technical condition.
The Lapidus (or distance) zeta function of $(A,\O)$ is then defined as the Lebesgue integral
\begin{equation}\label{zeta_dist}
\zeta_{A,\O}(s):=\int_{\O}d(x,A)^{s-N}\di x,
\end{equation}
for all $s\in\Ce$ such that $\re s$ is sufficiently large, where $d(x,A)$ denotes the Euclidean distance from $x$ to $A$.
Its main property is that the {\em abscissa of convergence} $D(\zeta_{A,\O})$ of $\zeta_{A,\O}$ coincides with the upper box dimension of $(A,\O)$, i.e., $D(\zeta_{A,\O})=\ov{\dim}_B(A,\O)$.
In other words, the integral \eqref{zeta_dist} converges absolutely and defines a holomorphic function in the open half-plane $\{\re s>\ov{\dim}_B(A,\O)\}$.
For the study of relative fractal drums, their corresponding fractal zeta functions and the general higher-dimensional theory of complex dimensions see \cite{fzf,ra} along with the survey articles \cite{brezish,tabarz}.
This higher-dimensional theory generalizes the well known theory of geometric zeta functions for fractal strings and their complex dimensions developed by Michel L. Lapidus and his collaborators in the last two decades (see \cite{lapidusfrank12} and the relevant references therein).  

In the case when the set $A$ degenerates to the point at infinity, we will denote this new kind of relative fractal drum with $(\infty,\O)$.
In this case it is clear that the fractal properties of such a relative fractal drum will depend only on the set $\O$.
We will extend the notions of Minkowski content and box dimension for such relative fractal drums and define a new class of Lapidus zeta functions associated to them.
Furthermore, it will be shown that this new class of Lapidus zeta functions has analogous properties as in the case of ordinary relative fractal drums and hence, provides an analytic approach to the study of fractality of unbounded sets.

The motivation to study the fractal properties of unbounded sets comes from a variety of sources.
In particular, the notion of "unbounded" or "divergent" oscillations appears in problems in oscillation theory (see, e.g. \cite{Dzu,Karp}), automotive industry (see, e.g., \cite{She}), civil engineering (see, e.g, \cite{Pou}) and mathematical applications in biology (see, e.g., \cite{May}).
Unbounded (divergent) oscillations are oscillations the amplitude of which increases with time.
For instance, the oscillations of an airplane that has positive static stability but negative dynamic stability is an example of divergent oscillations that appears in aerodynamics (see, e.g. \cite{Dol}).

Furthermore, unbounded domains themselves are also interesting in the theory of elliptic partial differential equations.
More precisely, the question of solvability of the Dirichlet problem for quasilinear equations in unbounded domains is addressed in~\cite{Maz1} and~\cite[Section~15.8.1]{mazja}.
Also, unbounded domains can be found in other aspects of the theory of partial differential equations; see, for instance~\cite{An,Hur,La,Rab} and~\cite{VoGoLat}.

\section{Minkowski Content and Box Dimension of Unbo\-unded Sets at Infinity}\label{inf_box_def}

Let $\O$ be a Lebesgue measurable subset of the $N$-dimensional Euclidean space $\eR^N$ of finite Lebesgue measure, i.e., $|\O|<\infty$.
Firstly, we will introduce a new notation for the sake of brevity, namely,
\begin{equation}\label{kratica}
_t\O:=B_t(0)^c\cap\O,
\end{equation}
where $t>0$ and $B_t(0)^c$ denotes the complement of the open ball of radius $t$ centered at $0$.
We introduce the {\em tube function of $\O$ at infinity} by $t\mapsto|B_t(0)^c\cap\O|$ for $t>0$ where $|\cdot|$ denotes the $N$-dimensional Lebesgue measure and we will be interested in the asymptotic properties of this function when $t\to +\ty$.
Furthermore, for any real number $r$ we define the {\em upper $r$-dimensional Minkowski content} of $\O$ {\em at infinity}
\begin{equation}\label{upMinkinf}
{\ovb{\M}}^{r}(\ty,\O):=\limsup_{t\to+\infty}\frac{|{B_t(0)^c\cap\O}|}{t^{N+r}},
\end{equation}
and, analogously, by taking the lower limit in \eqref{upMinkinf} as $t\to +\infty$, we define the {\em lower $r$-dimensional Minkowski content} of $\O$ {\em at infinity} denoted by ${\unb{\M}}^{r}(\ty,\O)$.
It is easy to see that the above definition implies the existence of a unique $D\in\eR$ such that ${\ovb{\M}}^{r}(\ty,\O)=+\ty$ for $r<\ovb{D}$ and ${\ovb{\M}}^{r}(\ty,\O)=0$ for $r>\ovb{D}$ and similarly for the lower Minkowski content.
The value $\ovb{D}$ is called the {\em upper box dimension of $\O$ at infinity} and we denote it with ${{\ovb{\dim}}}_B(\ty,\O)$.
Similarly as in the case of ordinary relative fractal drums, we have
\begin{equation}\label{updiminf}
{{\ovb{\dim}}}_B(\ty,\O):=\sup\{r\in\eR:{\ovb{\M}}^{r}(\ty,\O)=+\ty\}=\inf\{r\in\eR:{\ovb{\M}}^{r}(\ty,\O)=0\}.
\end{equation}
Analogously, by using the lower Minkowski content of $\O$ at infinity, we define the {\em lower box dimension of $\O$ at infinity} and denote it by ${\unb{\dim}}_B(\ty,\O)$ and the analog of \eqref{updiminf} is also valid in this case.
Of course, if the upper and lower box dimensions coincide, we define the {\em box dimension} of $\O$ {\em at infinity} and denote it with $\dim_B(\ty,\O)$.

In the case when the upper and lower Minkowski content at infinity coincide we define the {\em $r$-dimensional Minkowski content} of $\O$ {\em at infinity} and denote it with $\M^{r}(\ty,\O)$.
Furthermore, in the case when 
$
0<{\unb{\M}}^{D}(\ty,\O)\leq{\ovb{\M}}^{D}(\ty,\O)<+\ty.
$
for some $D\in\eR$ (which implies that $D=\dim_B(\ty,\O)$), we say that $\O$ is {\em Minkowski nondegenerate at infinity}.
We say that $\O$ is {\em Minkowski measurable at infinity} if it is Minkowski nondegenerate at infinity and its lower and upper Minkowski content at infinity coincide.

\begin{proposition}\label{<=-N}
Let $\O$ be a Lebesgue measurable subset of $\eR^N$ of finite Lebesgue measure.
Then ${\unb{{\dim}}}_B(\ty,\O)\leq{{\ovb{\dim}}}_B(\ty,\O)\leq-N$, i.e., the upper and lower box dimensions of $\O$ at infinity are always negative, that is, less than or equal to $-N$.
\end{proposition}

\begin{proof}
From the definitions \eqref{upMinkinf} and \eqref{updiminf} and the fact that $|\O|<\infty$ we have that $|B_t(0)^c\cap\O|\to 0$ as $t\to+\ty$ which implies that if $N+r>0$, then ${\ovb{\M}}^{r}(\ty,\O)=0$.
\end{proof}

\begin{remark}\label{=-N}
Intuitively the conclusion of Proposition~\ref{<=-N} is expected, since $\O$ having finite Lebesgue measure implies that it must have a certain flatness property relative to infinity.
(Compare with the notion of flatness introduced in~\cite{fzf}.)
Furthermore, if ${\ovb{\dim}}_B(\ty,\O)=-N$, then it follows from the definition that ${\ovb{\M}}^{-N}(\ty,\O)=0$ and, consequently, $\O$ must be Minkowski degenerate at infinity.
\end{remark}


The next two results about the monotonicity are simple consequences of the definitions involved. 

\begin{lemma}\label{subset_inf}
Let $\O_1\subseteq\O_2\subseteq\eR^N$ be two Lebesgue measurable sets and $|\O_2|<\ty$.
Then for any real number $r$ we have that
$
{\ovb{\M}}^r(\ty,\O_1)\leq{\ovb{\M}}^r(\ty,\O_2)$ and ${\unb{\M}}^r(\ty,\O_1)\leq{\unb{\M}}^r(\ty,\O_2)
$
\end{lemma}


\begin{corollary}\label{subset_inf_dim}
Let $\O_1\subseteq\O_2\subseteq\eR^N$ be two Lebesgue measurable sets with $|\O_2|<\ty$.
Then,
$
{\ovb{\dim}}_B(\ty,\O_1)\leq{\ovb{\dim}}_B(\ty,\O_2)$ and ${\unb{\dim}}_B(\ty,\O_1)\leq{\unb{\dim}}_B(\ty,\O_2).
$
\end{corollary}

Let us now take a look at a few examples.

\begin{definition}\label{Omega(a,b)}
Let $\a>0$ and $\b>1$ be fixed and define
$
a_j:=j^{\a}$, $l_j:=j^{-\b}$ and $b_j:=a_j+l_j.
$
We define
\begin{equation}\label{omega_a_b}
\O(\a,\b):=\bigcup_{j=1}^{\ty}I_j\subseteq\eR,
\end{equation}
that is, as a union of countably many intervals $I_j:=(a_j,b_j)$. 
\end{definition}

\begin{proposition}\label{1dimexample}
For the set $\O(\a,\b)$ defined by~\eqref{omega_a_b} we have that
\begin{equation}\label{dim(a,b)}
D:=\dim_B(\ty,\O(\a,\b))=\frac{1-(\a+\b)}{\a}\quad\textrm{ and }\quad\M^{D}(\ty,\O(\a,\b))=\frac{1}{\b-1}.
\end{equation}
\end{proposition}

\begin{proof}
Firstly, we observe that for $j$ large enough the intervals $I_j$ become disjoint, i.e., $j^{-\b}<(j+1)^{\a}-j^{\a}$.
As we see, $\O(\a,\b)$ is a union of intervals that ``escape'' to infinity and $|\O(\a,\b)|\leq\sum_{j=1}^{\ty}j^{-\b}<\infty$. Let us compute the box dimension and Minkowski content of $\O(\a,\b)$ at infinity. For $t>0$ let $j_0$ be such that for every $j>j_0$ it holds that $a_j>t$, that is, $j_0=\lfloor t^{1/\a}\rfloor$.
Now we fix $t$ large enough so that the intervals $I_{j}$ are disjoint for $j\geq j_0$.
From this, we have
\begin{equation}
|B_t(0)^c\cap{\O(\a,\b)}|=\sum_{j>j_0}j^{-\b}+\chi_{\O}(t)(b_{j_0}-t),
\end{equation}
with $\chi_{\O}$ being the characteristic function of $\O$. This implies the following estimate
\begin{equation}
\sum_{j>j_0}j^{-\b}\leq|B_t(0)^c\cap{\O(\a,\b)}|\leq\sum_{j>j_0}j^{-\b}+j_0^{-\b}.
\end{equation} 
Furthermore, using the integral criterion 
$
\int_{j_0+1}^{+\ty}x^{-\b}\di x\leq\sum_{j>j_0}j^{-\b}\leq(j_0+1)^{-\b}+\int_{j_0+1}^{+\ty}x^{-\b}\di x
$
for estimating the sum, we have
$$
\frac{1}{\b-1}(j_0+1)^{1-\b}\leq|B_t(0)^c\cap{\O(\a,\b)}|\leq\frac{1}{\b-1}(j_0+1)^{1-\b}+(j_0+1)^{-\b}+j_0^{-\b}.
$$
Finally, by using the fact that $t^{1/\a}-1\leq j_0+1\leq t^{1/\a}+1$, we conclude that
$$
\frac{1}{\b-1}(t^{\frac{1}{\a}}+1)^{1-\b}\leq|B_t(0)^c\cap{\O(\a,\b)}|\leq\frac{1}{\b-1}(t^{\frac{1}{\a}}-1)^{1-\b}+2(t^{\frac{1}{\a}}-2)^{-\b}
$$ 
which implies that $\M^{r}(\ty,\O(\a,\b))$ is different from $0$ and $+\ty$ if and only if $r+1=(1-\b)/\a$, i.e., if \eqref{dim(a,b)} holds.
\end{proof}

As we can see, the Minkowski content in the above case depends only on the parameter $\b$, i.e., the rate at which $\O(\a,\b)$ ``escapes'' to infinity is not relevant for it.
Furthermore, by changing the values of parameters $\a$ and $\b$, we can obtain any prescribed value in $(-\ty,-1)$ for $\dim_B(\ty,\O(\a,\b))$.
Moreover, we have that $\dim_B(\ty,\O(\a,\b))\to -\ty$ and $\M^D(\ty,\O(\a,\b))\to 0$ as $\b\to+\ty$.

\begin{proposition}\label{standardniprim}
For $\a>1$ let $\O:=\{(x,y)\in\eR^2\,:\,x>1,\ 0<y<x^{-\a}\}$. Then we have that
\begin{equation}\label{dimstandardniprim}
D:=\dim_B(\ty,\O)=-1-\a\quad \mathrm{ and }\quad \M^{D}(\ty,\O)=\frac{1}{\a-1}.
\end{equation}
\end{proposition}

\begin{proof}
Let $t>1$ and let $x(t)$ be such that
\begin{equation}\label{xxx}
x(t)^2+x(t)^{-2\a}=t^2.
\end{equation}
Then we have
$
\int_t^{+\ty}x^{-\a}\di x\leq|B_t(0)^c\cap\O|\leq\int_{x(t)}^{+\ty}x^{-\a}\di x
$
which implies that
$
\frac{1}{1-\a}\leq\frac{|B_t(0)^c\cap\O|}{t^{1-\a}}\leq\frac{1}{1-\a}\left(\frac{x(t)}{t}\right)^{1-\a}.
$
Furthermore, from \eqref{xxx} we have that
$
\frac{x(t)}{t}=(1+x(t)^{-2(\a+1)})^{-\frac{1}{2}}\to 1,
$
as $t\to+\infty$, and we conclude that~\eqref{dimstandardniprim} holds. 
\end{proof}

\begin{remark}
Note that $\dim_B(\ty,\O)\to -\ty$ and $\M^{D}(\ty,\O)\to 0$ as $\alpha\to +\ty$.
\end{remark}

Next we will prove a useful lemma which states that the box dimension and Minkowski measurability at infinity are independent on the choice of the norm on $\eR^N$ in a sense that we can replace the ball $B_t(0)$ in the definition of the Minkowski content at infinity with a ball in any other norm on $\eR^N$.
More precisely, let $\|\cdot\|$ be another norm on $\eR^N$.
We denote by $K_t(0)$ the open ball of radius $t$ around $0$ in the new norm; define the associated upper Minkowski content at infinity
$$
{\ovb{\mathcal{N}}}^{r}(\ty,\O):=\limsup_{t\to+\ty}\frac{|K_t(0)^c\cap\O|}{t^{N+r}}
$$
and analogously, ${\unb{\mathcal{N}}}^{r}(\ty,\O)$ and ${\mathcal{N}}^{r}(\ty,\O)$.   

\begin{lemma}\label{eqnorms}
Let $\O\subseteq\eR^N$ with $|\O|<\infty$ and assume that two norms, $|\cdot|$ and $\|\cdot\|$, are given on $\eR^N$, i.e., there are $a,b>0$ such that $a|\cdot|\leq\|\cdot\|\leq b|\cdot|$. Then, for any $r\in\eR$ we have
\begin{equation}
a^{-(N+r)}{\ovb{\mathcal{M}}}^{r}(\ty,\O)\leq{\ovb{\mathcal{N}}}^{r}(\ty,\O)\leq b^{-(N+r)}{\ovb{\mathcal{M}}}^{r}(\ty,\O),
\end{equation}
and analogously for the corresponding lower Minkowski contents. 
\end{lemma}

\begin{proof}
From $a|x|\leq\|x\|\leq b|x|$ we have that $B_{t/b}(0)\subseteq K_t(0)\subseteq B_{t/a}(0)$ for any $t>0$ and, consequently,
$$
a^{-(N+r)}\frac{|B_{t/a}(0)^c\cap\O|}{\left(\frac{t}{a}\right)^{N+r}}\leq\frac{|K_t(0)^c\cap\O|}{t^{N+r}}\leq b^{-(N+r)}\frac{|B_{t/b}(0)^c\cap\O|}{\left(\frac{t}{b}\right)^{N+r}}.
$$
Taking the upper limit as $t\to+\ty$, we obtain the first statement of the lemma.
The second one is obtained by taking the lower limit instead of the upper.
\end{proof}

\begin{corollary}
Let $\O$ be an arbitrary Lebesgue measurable subset of $\eR^N$ with finite $N$-dimensional Lebesgue measure.
Then

$(a)$ The upper and lower box dimensions of $\O$ at infinity do not depend on the choice of the norm on $\eR^N$ in which we measure the neighborhood of infinity.
	
$(b)$ The Minkowski nondegeneracy of $\O$ is independent of the choice of the norm on $\eR^N$ in which we measure the neighborhood of infinity.
\end{corollary}


There are special cases when we even get the same values for the Minkowski contents for different norms on $\eR^N$.
One of these cases is addressed in the next lemma which will prove to be useful in some of the future calculations.
It can easily be generalized to the $N$-dimensional case but we will need it only in the case of $\eR^2$. 

\begin{lemma}\label{infty_norm}
Let $\O\subseteq\eR^2$ with $|\O|<\infty$ such that $\O$ is a subset of a horizontal $($vertical$)$ strip of finite width.
Let $K_t(0)$ be an open ball in the $|\cdot|$-norm of radius $t>0$ with center at the origin and $r$ a real number.
Then, we have that
$
{{\ovb{\mathcal{M}}}}^{r}(\ty,\O)={{\ovb{\mathcal{N}}}}^{r}(\ty,\O)$ and ${\unb{{\mathcal{M}}}}^{r}(\ty,\O)={\unb{{\mathcal{N}}}}^{r}(\ty,\O).
$
\end{lemma}

\begin{proof}
Without loss of generality we will assume that the set $\O$ is contained in the horizontal half-strip $\{(x,y)\,:\,x\geq 0,\ 0\leq y\leq d\}$.
Then, for $t\geq d$ we have that $|K_{\sqrt{t^2-d^2}}(0)^c\cap\O|\leq|B_t(0)\cap\O|\leq|K_t(0)^c\cap\O|$ and consequently for $r\in\eR$
$$
\frac{(\sqrt{t^2-d^2})^{N+r}}{t^{N+r}}\frac{|K_{\sqrt{t^2-d^2}}(0)^c\cap\O|}{(\sqrt{t^2-d^2})^{N+r}}\leq\frac{|B_t(0)\cap\O|}{t^{N+r}}\leq\frac{|K_t(0)^c\cap\O|}{t^{N+r}}.
$$
Taking the upper and lower limits as $t\to +\ty$ completes the proof. 
\end{proof}
In the next example we will show that the value $\dim_B(\ty,\O)=-\infty$ can be achieved.
\begin{example}
Let $\O:=\{(x,y)\in\eR^2\,:\,x>1,\ 0<y<\mathrm{e}^{-x}\}$ and let us calculate the box dimension of $\O$ at infinity using the $|\cdot|_{\ty}$-ball in $\eR^2$:
$
|K_t(0)^c\cap\O|=\int_t^{+\ty}\mathrm{e}^{-x}\di x=\mathrm{e}^{-t}.
$
Consequently, we have that
$
\frac{|K_t(0)^c\cap\O|}{t^{2+r}}=\frac{\mathrm{e}^{-t}}{t^{2+r}}\to 0
$
when $t\to+\ty$ for every $r\in\eR$ and therefore $\dim_B(\ty,\O)=-\ty$.
\end{example}

\begin{remark}
From now on, we will always implicitly assume that $\ovb{\dim}_B(\ty,\O)>-\infty$ when dealing with relative fractal drums of the type $(\ty,\O)$ (unless stated otherwise).
\end{remark}

As we have shown in Proposition~\ref{<=-N}, the upper box dimension of any subset of the plane of finite Lebesgue measure does not exceed $-2$.
The next proposition will show that the value $-2$ can be achieved and it can be easily adapted for constructing a subset $\O$ of $\eR^N$ with finite Lebesgue measure such that ${\ovb{\dim}}_B(\ty,\O)=-N$. 

\begin{proposition}\label{prim_D=-2}
There exists a Lebesgue measurable subset $\O\subseteq\eR^2$ with $|\O|<\ty$ such that
\begin{equation}\label{-2_dim}
\dim_B(\ty,\O)=-2\quad \textrm{ and }\quad \M^{-2}(\ty,\O)=0.
\end{equation} 
\end{proposition}

\begin{proof}
Let $\a_k:=1+1/k$ for $k\geq 1$ and we define
$$
\widetilde{\O}_k:=\left\{(x,y)\in\eR^2\,:\,x>1,\ 0<y<\frac{2^{-k}}{k}x^{-\a_k}\right\}.
$$
We will ``stack'' the sets $\widetilde{\O}_k$ on top of each other.
In order to do so, we define $\O_k$ to be an $S_k$-translated image of $\widetilde{\O}_k$ along the $y$-axis where
$
S_k:=\sum_{j=1}^k{2^{-j}}{j}^{-1}
$
and define $\O:=\cup_{k\geq 1}\O_k$.
We observe that $\O$ is contained in the horizontal strip of finite height
$
\{(x,y)\in\eR^2\,:\,1/2\leq y\leq S\},
$
where $S:=\lim_{k\to\ty}S_k=\log 2$.
Furthermore, we have that
$$
|\O_k|=|\widetilde{\O}_k|=\frac{2^{-k}}{k}\int_{1}^{+\ty}x^{-1-\frac{1}{k}}\di x=\frac{2^{-k}}{k}\cdot k=2^{-k}
$$
so that $|\O|=\sum_{k=1}^{\ty}2^{-k}=1$.
Using the same calculation as in Proposition~\ref{standardniprim} yields
$$
D_k:=\dim_B(\ty,\O_k)=-1-\a_k=-2-\frac{1}{k}\quad \textrm{ and }\quad \M^{D_k}(\ty,\O_k)=2^{-k}.
$$
Finally, by using Corollary~\ref{subset_inf_dim} we have that $-2\geq\dim_B(\ty,\O)\geq D_k$ for every $k\geq 1$ which implies~\eqref{-2_dim}.
\end{proof}

\section{Holomorphicity of Lapidus Zeta Functions at Infinity}\label{holo_ty_La}

Let $\O\subseteq\eR^N$ be a measurable set with $|\O|<\ty$. We define the  {\em Lapidus zeta function} of $\O$ {\em at infinity} by the Lebesgue integral

\begin{equation}\label{infzeta}
\zeta_{\ty,\O}(s):=\int_{{B_T(0)^c\cap\O}}|x|^{-s-N}\di x,
\end{equation}
for a fixed $T>0$ and $s$ in $\Ce$ with $\re s$ sufficiently large.
We will also call this zeta function the {\em distance zeta function of $\O$ at infinity} and use the two notions interchangeably.
From now on, our main goal will be to show that this new zeta function has analogous properties as the distance zeta function for relative fractal drums studied in \cite{fzf,ra}.
First of all, the dependence of the distance zeta function at infinity on $T>0$ is inessential in the sense that for $0<T_1<T_2$ the difference
$$
\zeta_{\ty,\O}(s;T_1)-\zeta_{\ty,\O}(s;T_2)=\int_{B_{T_1,T_2}(0)\cap\O}|x|^{-s-N}\di x,
$$
with
\begin{equation}
B_{a,b}(0):=\{x\in\eR^N\,:\, a<|x|<b\},
\end{equation}
is an entire function of $s$.
Indeed, since $T_1\leq|x|\leq T_2$ for $x\in E$ this will follow from Theorem \ref{an2}$(c)$ with $E:=B_{T_1,T_2}(0)\cap\O$, $\varphi(x):=|x|$ and $\di\mu(x):=|x|^{-N}\di x$ in the notation of that theorem.
Therefore, from now on, we will emphasize the dependence of the Lapidus zeta function of $\O$ at infinity on $T$ and write $\zeta_{\ty,\O}(s;T)$ only when it is explicitly needed.
Also note that if $\O$ is bounded, then for $T$ sufficiently large, we have that $\zeta_{\ty,\O}(s;T)\equiv 0$.

The definition of the Lapidus zeta function of $\O$ at infinity is, as we will demonstrate immediately, closely related to the distance zeta function of a certain relative fractal drum.
This relative fractal drum is actually the image of $(\ty,\O)$ under the geometric inversion in $\eR^N$, i.e., it is equal to $(\mathbf 0,\Phi(\O))$\footnote{We should actually write $(\{\mathbf 0\},\Phi(\O))$ here, but we will always abuse notation in this way for a relative fractal drum $(A,\O)$ when the set $A$ consists of a single point.}, where
\begin{equation}\label{def_inv}
\Phi(x):=\frac{x}{|x|^2}
\end{equation}
and $\mathbf{0}$ is the origin.
To derive the mentioned relation we will need to compute the Jacobian of the geometric inversion and use the change of variables formula for the Lebesgue integral.
To compute the Jacobian we will use the well-known Matrix determinant lemma (see, e.g., \cite{Har}) which we state here for the sake of exposition.

\begin{lemma}[Matrix determinant lemma]\label{MDL}
Let $\mathbf{A}$ be an invertible matrix and $\mathbf{u}$, $\mathbf{v}$ column vectors. Then we have that
$
\det(\mathbf{A}+\mathbf{u}\otimes\mathbf{v})=(1+\mathbf{v}^{\tau}\mathbf{A}^{-1}\mathbf{u})\det\mathbf{A},
$
where
$
\mathbf{u}\otimes\mathbf{v}:=\mathbf{u}\mathbf{v}^{\tau}
$
and $\tau$ denotes the transpose operator.
\end{lemma}

\begin{lemma}\label{jacobian}
Let $\Phi(x):=x/|x|^2$ be the geometric inversion on $\eR^N$.
Then for the Jacobian of $\Phi$ we have$:$
$
\det\frac{\pa\Phi}{\pa x}=-|x|^{-2N}.
$
\end{lemma} 

\begin{proof}
With $x=(x_1,\ldots,x_N)$ and $\delta_{ij}$ the Kronecker delta we have that
\begin{equation}
\begin{aligned}
\left(\frac{\pa\Phi}{\pa x}\right)_{ij}&=\frac{\pa\Phi_i}{\pa x_j}=\frac{\delta_{ij}}{|x|^2}-\frac{2x_ix_j}{|x|^4}\\
\end{aligned}
\end{equation}
and consequently
\begin{equation}\label{phi_der}
\frac{\pa\Phi}{\pa x}=\frac{1}{|x|^4}(|x|^2\mathbf{I}-2\,\mathbf{x}\otimes\mathbf{x}),
\end{equation}
where $\mathbf{x}:=[x_1,\ldots,x_N]^{\tau}$ and $\mathbf{I}$ is the identity matrix.
Now we can apply the matrix determinant lemma with $\mathbf{A}:=|x|^2\mathbf{I}$, $\mathbf{u}:=-2\mathbf{x}$ and $\mathbf{v}:=\mathbf{x}$ from which we obtain
$$
\begin{aligned}
\det\frac{\pa\Phi}{\pa{x}}&\!=\!\frac{1}{|x|^{4N}}(1\!-\!2\mathbf{x}^{\tau}(|x|^2\mathbf{I})^{-1}\mathbf{x})\det(|x|^2\mathbf{I})\!=\!\frac{(1\!-\!2|x|^{-2}\mathbf{x}^{\tau}\mathbf{x})|x|^{2N}}{|x|^{4N}}\!=\!-|x|^{-2N}\!.\\
\end{aligned}
$$
\end{proof}

The next theorem will show that, from the point of view of the distance zeta functions, there is no difference between the unbounded relative fractal drum $(\ty,\O)$ and the relative fractal drum $(\mathbf{0},\Phi(\O))$ obtained from it by geometric inversion.

\begin{theorem}\label{zeta_inversion}
Let $\O$ be a Lebesgue measurable subset of $\eR^N$ of finite measure, $\mathbf{0}$ the origin and fix $T>0$.
Then we have
\begin{equation}
\zeta_{\ty,\O}(s;T)=\zeta_{\mathbf{0},\Phi(\O)}(s;1/T).
\end{equation}
\end{theorem}

\begin{proof}
Defining $y=\Phi^{-1}(x)$ and using Lemma~\ref{jacobian} this is a consequence of the change of variables formula once we observe the fact that $|y|=1/|x|$:
$$
\begin{aligned}
\zeta_{\ty,\O}(s;T)&=\int_{B_T(0)^c\cap\O}|x|^{-s-N}\di x=\int_{\Phi(B_T(0)^c\cap\O)}|y|^{s+N}|y|^{-2N}\di y\\
&=\int_{B_{1/T}(0)\cap\Phi(\O)}|y|^{s-N}\di y=\zeta_{\{\mathbf{0}\},\Phi(\O)}(s;1/T).
\end{aligned}
$$
\end{proof}

This result suggests that we can analyze fractal properties of $\O\subseteq\eR^N$ at infinity by analyzing the fractal properties of the `inverted' relative fractal drum $(\mathbf{0},\Phi(\O))$.
A similar approach (in the context of unbounded subsets of $\eR^N$) was made in~\cite{razuzup}.
Of course, in that approach, we can use results of~\cite{fzf} about relative fractal drums and relative distance (and tube) zeta functions.
On the other hand, we stress that in that case we are dealing with the usual relative box dimension of the inverted relative fractal drum, i.e., with $\dim_B(\mathbf{0},\Phi(\O))$ which is defined via the $r$-dimensional relative Minkowski content, namely, $\mathcal{M}^r(\mathbf{0},\Phi(\O))$.
However, it is not evident what are the relations between the ``classical'' relative box dimension (and Minkowski content) of the inverted relative fractal drum with the notions of box dimension and Minkowski content at infinity introduced in Section~\ref{inf_box_def}.
We will give an answer to this question in a future work as well as to the natural question about the effect of the one-point compactification on the fractal properties of unbounded sets at infinity as well as how to analyze fractal properties of unbounded sets of infinite measure at infinity.
(See also \cite{ra}.)

To prove the holomorphicity theorem, we will need the following proposition which complements~\cite[Lemma~3]{singl}.

\begin{proposition}\label{integralna_veza}
Let $\O\subseteq\eR^N$ be a Lebesgue measurable set with $|\O|<\ty$, $T>0$ and let $u\colon(T,+\ty)\to[0,+\ty)$ be a strictly monotone $C^1$ function. Then the following equality holds
\begin{equation}\label{int_vez}
\int_{{B_T(0)^c\cap\O}}u(|x|)\di x=u(T)|{B_T(0)^c\cap\O}|+\int_T^{+\ty}|B_t(0)^c\cap\O|u'(t)\di t.
\end{equation}
\end{proposition}

\begin{proof}
We will use a well-known fact (see, e.g., \cite[Theorem 1.15]{Mat}) that for a non-negative Borel function $f$ on a separable metric space $X$ the following identity holds

\begin{equation}\label{slojevita}
\int_Xf(x)\di x=\int_0^{\infty}|\{x\in X\,:\, f(x)\geq t\}|\di t.
\end{equation}

We let $f(x):=u(|x|)$, $X:=B_T(0)^c\cap{{\O}}$ and consider separately the cases of strictly decreasing and strictly increasing function $u$.

$(a)$ Let $u$ be strictly decreasing and $u(+\infty):=\lim_{\tau\to+\ty}u(\tau)$. For the set appearing on the right-hand side of \eqref{slojevita} we have
$$
A(t):=\{x\in {B_T(0)^c\cap\O}\,:\, u(|x|)\geq t\}=\{x\in {B_T(0)^c\cap\O}\,:\, |x|\leq u^{-1}(t)\}.
$$
For $0\leq t\leq u(+\ty)$ it is true that $u(|x|)\geq t$ for any $x\in\eR^N$ because $u(+\ty)=\min_{\tau\geq 0} u(\tau)$ and we have $A(t)={B_T(0)^c\cap\O}$.
Furthermore, if $u(+\ty)<t\leq u(T)$, it is clear that
$$
A(t)=({B_T(0)^c\cap\O})\setminus (B_{u^{-1}(t)}(0)^c\cap\O)=B_{T,u^{-1}(t)}(0)\cap\O.
$$
Finally, for $t> u(T)$ we have that $A(t)=\emptyset$ because $u(T)=\max_{\tau\geq 0} u(\tau)$ and using \eqref{slojevita} we get
$$
\begin{aligned}
\int_{{B_T(0)^c\cap\O}}u(|x|)\di x&=\int_0^{u(+\ty)}|{B_T(0)^c\cap\O}|\di t+\int_{u(+\ty)}^{u(T)}|B_{T,u^{-1}(t)}(0)\cap\O|\di t\\
&=u(+\ty)|{B_T(0)^c\cap\O}|+\int_{u(+\ty)}^{u(T)}|B_{T}(0)^c\cap\O|\di t\\
&\phantom{=}-\int_{u(+\ty)}^{u(T)}|B_{u^{-1}(t)}(0)^c\cap\O|\di t\\
&=u(T)|{B_T(0)^c\cap\O}|+\int_{T}^{+\ty}|B_{s}(0)^c\cap\O|u'(s)\di s,
\end{aligned}
$$
where we have introduced the new variable $s=u^{-1}(t)$ in the last equality.

$(b)$ Let now $u$ be a strictly increasing function and $u(+\ty):=\lim_{\tau\to+\ty}u(\tau)=\sup_{\tau\geq 0}u(\tau)\in(0,+\ty]$.
In this case we have
$$
A(t):=\{x\in {B_T(0)^c\cap\O}\,:\, u(|x|)\geq t\}=\{x\in {B_T(0)^c\cap\O}\,:\, |x|\geq u^{-1}(t)\}.
$$
For $0\leq t\leq u(T)$ we have that $u(|x|)\geq t$ for any $x\in\eR^N$ because $u(T)=\min_{\tau\geq 0} u(\tau)$ and we have $A(t)={B_T(0)^c\cap\O}$.
Furthermore, if $u(T)<t<u(+\ty)$ it is clear that $A(t)=B_{u^{-1}(t)}(0)^c\cap\O$, and for $t\geq u(+\ty)$ the set $A(t)$ is an empty set.
Altogether, we have
$$
\begin{aligned}
\int_{{B_T(0)^c\cap\O}}u(|x|)\di x&=\int_0^{u(T)}|{B_T(0)^c\cap\O}|\di t+\int_{u(T)}^{u(+\ty)}|B_{u^{-1}(t)}(0)^c\cap\O|\di t\\
&=u(T)|{B_T(0)^c\cap\O}|+\int_{T}^{+\ty}|B_{s}(0)^c\cap\O|u'(s)\di s
\end{aligned}
$$
where, again, we have introduced the new variable $s=u^{-1}(t)$ in the last equality.
This concludes the proof of the proposition.
\end{proof}

\begin{proposition}\label{cor_infint}
Let $\O\subseteq\eR^N$ be a measurable set with $|\O|<\ty$, $T>0$. Then for every $\sigma\in({\ovb{\dim}}_B(\ty,\O),+\ty)$, the following identity holds$:$
\begin{equation}\label{infzetaint}
\int_{_T\O}|x|^{-\s-N}\di x=T^{-\s-N}|_T\O|-(\s+N)\int_T^{+\ty}t^{-\s-N-1}|_t\O|\di t.
\end{equation}
Furthermore, the above integrals are finite for such $\sigma$. 
\end{proposition}

\begin{proof}
The proposition is a direct consequence of Proposition~\ref{integralna_veza} with $u(t):=t^{-\s-N}$ when $\s\neq -N$ and for $\s=-N$ the equation~\eqref{infzetaint} is trivially fulfilled.
Namely, let us fix $\sigma_1\in({\ovb{\dim}}_B(\ty,\O),\sigma)$. Then for $T$ large enough we have that for a constant $M>0$ we have
$
|B_t(0)^c\cap\O|\leq Mt^{\sigma_1+N}
$
for every $t>T$.
From this we get that
$$
\int_T^{+\ty}t^{-\s-N-1}|B_t(0)^c\cap\O|\di t\leq M\int_T^{+\ty}t^{-\s-N-1}t^{\sigma_1+N}\di t=M\int_T^{+\ty}t^{\s_1-\s-1}\di t
$$
and the last integral above is finite because $\s_1-\s-1<-1$.
\end{proof}

In order to prove the holomorphicity theorem we will need the following theorem which we cite from \cite{fzf} along with its proof for the sake of exposition. 

\begin{theorem}[Cited from {\cite[Theorem 2.1.44]{fzf}}]\label{an2}
Let $(E,{\mathcal B}(E),\mu)$ be a measure space, where $E$ is a locally compact metrizable space, ${\mathcal B}(E)$ is the Borel $\s$-algebra of $E$, and $\mu$ is a positive or complex $($local$)$ measure, with total variation $($local$)$ measure denoted by $|\mu|$. Furthermore, let $\f:E\to(0,+\ty)$ be a measurable function. Then$:$
\bigskip

$(a)$ If $\f$ is essentially bounded $($that is, if there exists $C>0$ such that $\f(t)\le C$ for $|\mu|$-a.e.\ $t\in E)$,
and if there exists $\s\in\eR$ such that $\int_E\f(t)^\s \D|\mu|(t)$ $<\ty$, then 
\begin{equation}\label{Ffi}
F(s):=\int_E\f(t)^s\D\mu(t)
\end{equation}
is holomorphic on the right half-plane $\{\re s>\s\}$, and
$F'(s)=\int_E\f(t)^s\log\f(t)\,\D\mu(t)$ in that region.

$(b)$ If there exists $C>0$ such that $\f(t)\ge C$ for $|\mu|$-a.e.\ $t\in E$, and if there exists $\s\in \eR$ such that $\int_E\f(t)^{-\s} \D|\mu|(t)<\ty$,
then 
\begin{equation}
G(s):=\int_E\f(t)^{-s}\D\mu(t)
\end{equation}
is holomorphic on $\{\re s>\s\}$, and $G'(s)=-\int_E\f(t)^{-s}\log\f(t)\,\D\mu(t)$ in that region. 

$(c)$ Finally, if there exist positive constants $C_1$ and $C_2$ such that $C_1\le\f(t)\le C_2$ for $|\mu|$-a.e.\ $t\in E$, and there exists $\s\in\eR$ such that $\int_E\f(t)^{\s} \D|\mu|(t)<\ty$, then the Dirichlet-type integrals $F$ and $G$ in $(a)$ and $(b)$, respectively, are entire functions.
\end{theorem}

\begin{proof} We use \cite[Theorem B.4, page 295]{carlson} (see also \cite{Mattn}). In our case, $f(s,t):=\f(t)^s$, $Z:=\{\re s>\s\}$. Note that for any $\s_1>\s$, 
we have $\f(t)^{\s_1}\le\|\f\|_{\ty}^{\s_1-\s}\f(t)^\s$, so that $\f^\s\in L^1(|\mu|)$ implies that $\f^{\s_1}\in L^1(|\mu|)$. In particular, 
since $|f(s,t)|=\f(t)^{\re s}$, it follows that
$f(s,t)=\f(t)^s\in L^1(|\mu|)$ for all $s\in\Ce$ such that $\re s>\s$.

Let $K$ be a compact subset of $Z=\{\re s>\s\}$.
Since 
\begin{equation}\label{f(s,t)}
|f(s,t)|=\f(t)^{\re s}\le\|\f\|_\ty^{\re s-\s}\f(t)^\s,
\end{equation} 
 we have that 
$|f(s,t)|\le g_K(t):=C_K\f(t)^\s$
for all $s\in K$ and $|\mu|$-a.e.\ $t\in E$, where $C_K=\max_{s\in K}\|\f\|_\ty^{\re s-\s}$. 
This proves part $(a)$ of the theorem. 

Part $(b)$ follows from part $(a)$ applied to $\f(t)^{-1}$. 

Finally, part~$(c)$ follows similarly as in $(a)$, by noting that
\begin{equation}
|f(s,t)|=\f(t)^{\re s}\le
\max\{C_1^{\re s-\s},C_2^{\re s-\s}\}\f(t)^\s,
\end{equation} 
for every complex number $s$.
\end{proof}

Now we can state and prove the holomorphicity theorem for the Lapidus zeta function at infinity, but firstly we will introduce a new notation for the sake of brevity, namely,
\begin{equation}\label{kraticazasloj}
{_{a,b}\O}:=B_{a,b}(0)\cap\O.
\end{equation}

\begin{theorem}\label{analiticinf}
Let $\O$ be any Lebesgue measurable subset of $\eR^N$ of finite $N$-dimensional Lebesgue measure.
Assume that $T$ is a fixed positive number.
Then the following conclusions hold.

\noindent $(a)$ The abscissa of convergence of the Lapidus zeta function at infinity
\begin{equation}\label{inteqzeta}
\zeta_{\ty,\O}(s)=\int_{{B_T(0)^c\cap\O}}|x|^{-s-N}\di x
\end{equation}
is equal to the upper box dimension of $\O$ at infinity, i.e.,
\begin{equation}
D(\zeta_{\ty,\O})=\ovb{\dim}_B(\ty,\O).
\end{equation}
Consequently, $\zeta_{\ty,\O}$ is holomorphic on the half-plane $\{\re s>{\ovb{\dim}}_B(\ty,\O)\}$ and for every complex number $s$ in that half-plane we have that
\begin{equation}\label{zetainfder}
\zeta'_{\ty}(s,\O)=-\int_{{B_T(0)^c\cap\O}}|x|^{-s-N}\log|x|\di x.
\end{equation}


\noindent $(b)$ If $D=\dim_B(\ty,\O)$ exists and ${\unb{\M}}^{D}(\ty,\O)>0$, then $\zeta_{\ty,\O}(s)\to+\ty$ for $s\in\eR$ as $s\to D^+$.
\end{theorem}

\begin{proof}

$(a)$ If we let $\ovb{D}:={\ovb{\dim}}_B(\ty,\O)$, then from the definitions of the upper Minkowski content and of the upper box dimension at infinity we deduce that
$
\limsup_{t\to+\ty}\frac{|B_t(0)^c\cap\O|}{t^{N+\sigma}}=0
$
for every $\sigma>\ovb{D}$.
Now, let us fix $\sigma_1$ such that $\ovb{D}<\sigma_1<\sigma$ and take $T>1$ large enough, such that for a constant $M>0$ it holds that
$
|B_t(0)^c\cap\O|\leq Mt^{\sigma_1+N}$ for every $t>T.
$
Furthermore, we estimate $\zeta_{\ty,\O}(\sigma)$ in the following way
$$
\begin{aligned}
\zeta_{\ty,\O}(\sigma)&=\int_{B_T(0)^c\cap\O}|x|^{-\sigma-N}\di x=\sum_{k=1}^{\ty}\int_{_{T^k,T^{k+1}}\O}|x|^{-\sigma-N}\di x\\
&\leq\sum_{k=1}^{\ty}\max\left\{(T^k)^{-\sigma-N},(T^{k+1})^{-\sigma-N}\right\}|_{T^k,T^{k+1}}\O|\\
&\leq\max\left\{1,T^{-\sigma-N}\right\}\sum_{k=1}^{\ty}(T^k)^{-\sigma-N}M(T^k)^{\sigma_1+N}\\
&=M\max\left\{1,T^{-\sigma-N}\right\}\sum_{k=1}^{\ty}(T^{\sigma_1-\sigma})^k<\infty.
\end{aligned}
$$
The last inequality follows from the fact that $T>1$ and $\sigma_1-\sigma<0$.
We let now $E:={B_T(0)^c\cap\O}$, $\varphi(x):=|x|$ and $\di\mu(x):=|x|^{-N}\di x$ and note that $\varphi(x)\geq T>1$ for $x\in E$.
Part $(a)$ follows now directly from Theorem \ref{an2}$(b)$.

To conclude the proof that $\ovb{D}$ is the abscissa of convergence of $\zeta_{\ty,\O}$ we take $s\in(-\ty,\ovb{D})$ and use  Proposition~\ref{cor_infint}:
\begin{equation}
\begin{aligned}
I_T:=\int_{_T\O}|x|^{-s-N}\di x&=\frac{|_T\O|}{T^{s+N}}-(s+N)\int_{T}^{+\ty}t^{-s-N-1}|_t\O|\di t\geq \frac{|_T\O|}{T^{s+N}}.
\end{aligned}
\end{equation}
Now, we fix $\sigma$ such that $s<\sigma<\ovb{D}$.
From ${\ovb{\M}}^{\sigma}(\ty,\O)=+\ty$ we conclude that there exists a sequence $(t_k)_{k\geq 1}$ such that
$
C_k:=\frac{|{_{t_k}\O}|}{t_k^{N+\sigma}}\to+\ty$ when $t_k\to+\ty.
$
It is clear that the function $T\to I_T$ is nonincreasing and we have
\begin{equation}
\begin{aligned}
I_T\geq I_{t_k}\geq t_{k}^{-s-N}|{_{t_k}\O}|=t_k^{-s-N}t_k^{N+\sigma}C_k=C_kt_k^{\sigma-s}\to+\ty.
\end{aligned}
\end{equation}
Therefore, $I_T=+\ty$ for every $s<\ovb{D}$ which proves that $D(\zeta_{\ty,\O})=\ovb{D}$.

$(b)$ Let us assume now that $D=\dim_B(\ty,\O)$ exists, and ${\unb{\M}}^{D}(\ty,\O)>0$.
From Proposition~\ref{<=-N} we have that $D\leq -N$.
On the other hand, the condition ${\unb{\M}}^{D}(\ty,\O)>0$ and Remark~\ref{=-N} imply that $D\neq -N$.
Consequently, we may assume that $D<-N$. 
Furthermore, ${\unb{\M}}^{D}(\ty,\O)>0$ implies that there exists a constant $C>0$ such that for a sufficiently large $T$ we have that $|_t\O|\geq Ct^{N+D}$ for every $t>T$.
Hence, for $D<s<-N$ we have the following:
\begin{equation}
\begin{aligned}
\zeta_{\ty,\O}(s)&=\int_{B_T(0)^c\cap\O}|x|^{-s-N}\di x=T^{-s-N}|_T\O|-(s+N)\int_T^{+\ty}t^{-s-N-1}|_t\O|\di t\\
&\geq -(s+N)\int_T^{+\ty}t^{-s-N-1}|_t\O|\di t\geq -(s+N)C\int_T^{+\ty}t^{-s-N-1+N+D}\di t\\
&=-(s+N)C\frac{T^{D-s}}{s-D}\to+\ty,
\end{aligned} 
\end{equation}
when $s\to D^+$, and this proves part $(b)$.
\end{proof}

\begin{remark}
In the special case when ${\ovb{\dim}}_B(\ty,\O)=-N$ we have from the definition of the upper Minkowski content at infinity that
$
{\ovb{\M}}^{-N}(\ty,\O)=0$ and $\zeta_{\ty,\O}(-N)=|_T\O|.
$
This shows that the condition ${\unb{\M}}^{D}(\ty,\O)>0$ from part $(b)$ of Theorem~\ref{analiticinf} cannot be omitted in the general case.
\end{remark}

\begin{remark}
Similarly as in the case of standard relative fractal drums (see~\cite{fzf}), it is easy to see that Theorem~\ref{analiticinf} is still true if we replace the norm appearing in the definition of the distance zeta function at infinity with any other norm on $\eR^N$.
\end{remark}

Let us now revisit Propositions~\ref{1dimexample} and~\ref{standardniprim} from the previous section and compute the corresponding distance zeta functions at infinity.

\begin{proposition}\label{1dimexamplezeta}
Let $\O:=\O(\a,\b)$ be the set from Definition~\ref{Omega(a,b)}.
Then, for $T:=a_{j_0}$ large enough so that $_T\O$ is a countable union of disjoint intervals we have that
\begin{equation}\label{formula}
\zeta_{\ty,\O}(s;T)=\frac{1}{s}\sum_{j=j_0}^{\ty}(j^{-\a s}-(j^{\a}+j^{-\b})^{-s}).
\end{equation}
Furthermore, we have that
\begin{equation}
D(\zeta_{\ty,\O}(\,\cdot\,;T))=\frac{1-(\a+\b)}{\a}=\dim_{B}(\ty,\O)
\end{equation}
and $s=0$ is a removable singularity of $\zeta_{\ty,\O}(\,\cdot\,;T)$.
\end{proposition}

\begin{proof}
For the distance zeta function of $\O$ at infinity we have:
$
\zeta_{\ty,\O}(s;T)=\int_{B_T(0)^c\cap\O}x^{-s-1}\di x=\sum_{j=j_0}^{\ty}\int_{a_j}^{b_j}x^{-s-1}\di x
$
from which follows~\eqref{formula} after integrating. 
By setting $\s:=\re s$ and using the mean value theorem for integrals, we estimate
$$
|\zeta_{\ty,\O}(s;T)|\leq\sum_{j=j_0}^{\ty}\int_{a_j}^{b_j}x^{-\s-1}\di x=\sum_{j=j_0}^{\ty}c_j^{-\s-1}(b_j-a_j)
$$
for some $c_j\in(a_j,b_j)$ so that $c_j\asymp j^{\a}$ as $j\to +\ty$ which, in turn, implies that
$
\sum_{j=j_0}^{\ty}c_j^{-\s-1}(b_j-a_j)\asymp \sum_{j=j_0}^{\ty}j^{-\a(\s+1)}j^{-\b}.
$
The right-hand side is convergent if and only if $\sigma>\frac{1-(\a+\b)}{\a}$ from which we conclude by using~\eqref{dim(a,b)} that $D(\zeta_{\ty,\O}(\,\cdot\,;T))=\frac{1-(\a+\b)}{\a}=\dim_B(\ty,\O)$, which is in accord with Theorem~\ref{analiticinf}. 

\end{proof}

\begin{proposition}
Let $\O:=\{(x,y)\in\eR^2\,:\,x>1,\ 0<y<x^{-\a}\}$ for $\a>1$. Then for the distance zeta function of $\O$ at infinity calculated using the $|\cdot|_{\ty}$ norm on $\eR^2$ we have
$$
\zeta_{\ty,\O}(s;1;|\cdot|_{\ty})=\frac{1}{s+\a+1}.
$$
It is meromorphic on $\Ce$ with a single simple pole at $s=-1-\a$.
In particular, $\dim_B(\ty,\O)=-1-\a$.
\end{proposition}

\begin{proof}
Let us compute the distance zeta function of $\O$ at infinity:
$$
\zeta_{\ty,\O}(s;1;|\cdot|_{\ty})=\int_{_1\O}|(x,y)|_{\ty}^{-s-2}\di x\mathrm{d}y=\int_{1}^{+\ty}\di x\int_{0}^{x^{-\a}}x^{-s-2}\di y=\frac{1}{s+\a+1}.
$$
The last equation holds if and only if $\re s>-1-\a$.
From this and~\eqref{dimstandardniprim}, we conclude that $D(\zeta_{\ty,\O}(\,\cdot\,;|\cdot|_{\ty}))=-1-\a=\dim_B(\ty,\O)$ which is, of course, in accord with Theorem~\ref{analiticinf}.
Moreover, the distance zeta function $\zeta_{\ty,\O}(\,\cdot\,;|\cdot|_{\ty})$ of $\O$ at infinity can be meromorphically extended to the whole complex plane with a single simple pole at $s=D$.
\end{proof}

Revisiting Proposition~\ref{prim_D=-2} will show that the conditions of Theorem~\ref{analiticinf} cannot be relaxed.

\begin{proposition}\label{primD=-2_zeta}
Let $\O$ be as in Proposition~\ref{prim_D=-2}. 
Then for the corresponding Lapidus zeta function at infinity calculated via the $|\cdot|_{\ty}$-norm on $\eR^2$ we have
\begin{equation}
\zeta_{\ty,\O}(s;|\cdot|_{\ty})=\sum_{k=1}^{\ty}\frac{2^{-k}}{k(s+2+\frac{1}{k})}.
\end{equation}
Furthermore, we also have that
\begin{equation}\label{eq1}
D(\zeta_{\ty,\O}(\,\cdot\,;|\cdot|_{\ty}))=\dim_B(\ty,\O)=-2
\end{equation}
and
$
\zeta_{\ty,\O}(-2;|\cdot|_{\ty})=|\O|=1.
$
Moreover, $\zeta_{\ty,\O}(\,\cdot\,;|\cdot|_{\ty})$ is holomorphic on the set
\begin{equation}\label{eq2}
\Ce\setminus(\{-2\}\cup\{-2-1/k\,:\,k\geq 1\})
\end{equation}
and $s=-2$ is an accumulation point of its simple poles. Finally, for the residues of $\zeta_{\ty,\O}(\,\cdot\,;|\cdot|_{\ty})$ we have that
$
\res\left(\zeta_{\ty,\O}\left(\,\cdot\,;|\cdot|_{\ty}\right),-2-\frac{1}{k}\right)=\frac{2^{-k}}{k}
$
for every $k\geq 1$.
\end{proposition}

\begin{proof}
Let us calculate the distance zeta function at infinity using the $|\cdot|_{\ty}$ norm on $\eR^N$.
For $T=1>\log 2$ we have that $|(x,y)|_{\ty}=x$ for $(x,y)\in{_1\O}$ and consequently
$$
\begin{aligned}
\zeta_{\ty,\O}(s;1;|\cdot|_{\ty})&=\int_{\O}|(x,y)|_{\ty}^{-s-2}\di x\mathrm{d}y=\sum_{k=1}^{\ty}\int_{\O_k}|(x,y)|_{\ty}^{-s-2}\di x\mathrm{d}y\\
&=\sum_{k=1}^{\ty}\int_{\O_k}x^{-s-2}\di x\mathrm{d}y=\sum_{k=1}^{\ty}\int_{\widetilde{\O}_k}x^{-s-2}\di x\mathrm{d}y\\
&=\sum_{k=1}^{\ty}\int_{1}^{+\ty}\di x\int_{0}^{\frac{2^{-k}}{k}x^{-\a_k}}x^{-s-2}\di y=\sum_{k=1}^{\ty}\frac{2^{-k}}{k}\int_{1}^{+\ty}x^{-s-3-\frac{1}{k}}\di x\\
&=\sum_{k=1}^{\ty}\frac{2^{-k}}{k(s+2+\frac{1}{k})}.
\end{aligned}
$$
The last equation above is valid if and only if $\re s>-2-1/k$ for every $k\geq 1$.
Furthermore, by using the Weierstrass $M$-test we have that the last sum appearing above defines a holomorphic function on $\Ce\setminus(\{-2\}\cup\{-2-1/k\,:\,k\geq 1\})$, which implies that $D(\zeta_{\ty,\O}(\,\cdot\,;|\cdot|_{\ty}))=-2$.
On the other hand, by direct computation we have that $\zeta_{\ty,\O}(-2;|\cdot|_{\ty})=|\O|=1$, but the zeta function cannot be even meromorphically extended to a neighborhood of $s=-2$.
This follows from the fact that for $\re s>-2$ we have that
$
\zeta_{\ty,\O}(s;|\cdot|_{\ty})=\sum_{k=1}^{\ty}\frac{2^{-k}}{k}z_k(s),
$
where the functions $z_k$ are meromorphic on $\Ce$ with simple poles at $s_k=-2-1/k$.
Furthermore, the above sum converges uniformly on compact subsets of $\Ce\setminus\{s_k\,:\,k\geq 1\}$, i.e., it defines a holomorphic function on that set, but it has an accumulation of simple poles at $s=-2$, and by the principle of analytic continuation, the same is true for $\zeta_{\ty,\O}(\,\cdot\,;|\cdot|_{\ty})$.
In other words, $D(\zeta_{\ty,\O}(\,\cdot\,;|\cdot|_{\ty}))=-2$ and this, in turn, is equal to $\dim_B(\ty,\O)$ according to~\eqref{-2_dim}.
\end{proof}

\begin{remark}
Although Proposition~\ref{primD=-2_zeta} is stated in terms of the distance zeta function calculated via the $|\cdot|_{\ty}$-norm, Proposition~\ref{euc_ty} below will guarantee that the difference $\zeta_{\ty,\O}(\,\cdot\,;|\,\cdot\,|_{\ty})-\zeta_{\ty,\O}$ is holomorphic at least on the half-plane $\{\re s>-4\}$.
From this we conclude that~\eqref{eq1} is also true for $\zeta_{\ty,\O}$, $\zeta_{\ty,\O}(-2)=1$, and $\zeta_{\ty,\O}$ is holomorphic (at least) on the set
$
\{\re s>-4\}\setminus(\{-2\}\cup\{-2-1/k\,:\,k\geq 1\})
$
with $s=-2$ being an accumulation point of its simple poles.
\end{remark}

\section{Residues of the Lapidus Zeta Function at Infinity}\label{ty_res_sec}

In this section we will derive results which relate the the upper and lower Minkowski content of $(\ty,\O)$ with the residue of the distance zeta function at infinity at $s=\dim_B(\ty,\O)$.

\begin{theorem}\label{resdistinf}
Let $\O\subseteq\eR^N$ be such that $|\O|<\ty$ and $\dim_B(\ty,\O)=D<-N$, $0<{\unb{\M}}^{D}(\ty,\O)\leq{\ovb{\M}}^{D}(\ty,\O)<\ty$.
If $\zeta_{\ty,\O}$ has a meromorphic continuation to a neighborhood of $s=D$, then $D$ is a simple pole and it holds that
\begin{equation}\label{mink_res_inf}
-(N+D){\unb{\M}}^{D}(\ty,\O)\leq\res(\zeta_{\ty,\O},D)\leq -(N+D){\ovb{\M}}^{D}(\ty,\O).
\end{equation}
Moreover, if $\O$ is Minkowski measurable at infinity, then we have
\begin{equation}
\res(\zeta_{\ty,\O},D)=-(N+D)\M^{D}(\ty,\O).
\end{equation}
\end{theorem}

\begin{proof}
Firstly, using the fact that ${\unb{\M}}^{D}(\ty,\O)>0$ we can apply part $(c)$ of Theorem~\ref{analiticinf} to get that $\zeta_{\ty,\O}(s)\to+\ty$ as $\eR\ni s\to D^+$.
In fact, by looking at the proof of part $(c)$ of Theorem~\ref{analiticinf} we can see that $s=D$ is a singularity of $\zeta_{\ty,\O}$ that is at least a simple pole.
It remains to show that the order of this pole is not greater than one.
Let us define
$C_T:=\sup_{t\geq T}\frac{|_t\O|}{t^{N+D}}$.
From ${\ovb{\M}}^{D}(\ty,\O)<+\ty$ we have that $C_T<+\ty$ for $T$ large enough.
Now, for $s\in\eR$ such that $D<s<-N$ by using Proposition~\ref{cor_infint} we have
\begin{equation}\label{res_compute}
\begin{aligned}
\zeta_{\ty,\O}(s)&=T^{-s-N}|_T\O|-(s+N)\int_T^{+\ty}t^{-s-N-1}|_t\O|\di t\\
&\leq T^{-s-N}C_TT^{N+D}-(s+N)\int_T^{+\ty}t^{-s-N-1}C_Tt^{N+D}\di t\\
&=C_TT^{D-s}-C_T(s+N)\int_{T}^{+\ty}t^{D-s-1}\di t\\
&=C_TT^{D-s}-C_T(s+N)\frac{T^{D-s}}{s-D}=-(N+D)C_T\frac{T^{D-s}}{s-D}.
\end{aligned}
\end{equation} 
This implies that $0\leq\zeta_{\ty,\O}(s)\leq C_1(s-D)^{-1}$ where $C_1>0$ is a constant independent of $s$ and $T$ and from this we conclude that $s=D$ is a pole of at most order one, i.e., it is a simple pole.
To compute the residue at $s=D$ we observe that its value is independent of $T$ because the difference $\zeta_{\ty,\O}(s;T_2)-\zeta_{\ty,\O}(s;T_1)$ is an entire function.
Furthermore, from~\eqref{res_compute} we have
$
(s-D)\zeta_{\ty,\O}(s)\leq -(N+D)C_TT^{D-s}
$ 
and taking limits on both sides as $s\to D^+$ yields
$
\res(\zeta_{\ty,\O},D)\leq -(N+D)C_T.
$
Finally, by taking the limit as $T\to+\ty$ we get $\res(\zeta_{\ty,\O},D)\leq -(N+D){\ovb{\M}}^{D}(\ty,\O)$.
The proof of the inequality involving the lower Minkowski content is completely analogous and this completes the proof.
\end{proof}

The next technical proposition is needed in order to establish a finer connection between the zeta function at infinity defined via the Euclidean norm and the one defined via the $|\cdot|_{\ty}$-norm.
It is very useful since the later zeta function can be calculated explicitly in the examples we are interested in.
The proof follows from a more general theorem (see \cite[Theorem 4.55]{ra}) which is proved by using the complex mean value theorem \cite[Theorem~2.2]{EvJa} and the theorem about complex differentiation under the integral sign (see, e.g., \cite{carlson,Mattn}).
Due to the technical nature we omit the proof here and refer the reader to \cite[Theorem 4.55 and Proposition 4.58]{ra} for the detailed proof.

\begin{proposition}\label{euc_ty}
Let $\O\subseteq\eR^N$ with $|\O|<\ty$ be such that it is contained in a cylinder
$
x_2^2+x_3^2+\cdots+x_N^2\leq C
$
for some constant $C>0$ where $x=(x_1,\ldots,x_N)$.
Furthermore, let $\ovb{D}:=\ovb{\dim}_B(\ty,\O)$ and $T>0$.
Then
\begin{equation}
\zeta_{\ty,\O}(s;T)-\int_{{B_T(0)^c\cap\O}}|x|_{\ty}^{-s-N}\di x
\end{equation}
is holomorphic on $($at least$)$ the half-plane $\{\re s>\ovb{D}-2\}$.

Furthermore, if any of the two distance zeta functions possesses a meromorphic extension to some open connected neighborhood $U$ of the critical line $\{\re s=\ovb{D}\}$, then the other one possesses a meromorphic extension to $($at least$)$ $V:=U\cap\{\re s>\ovb{D}-2\}$.
Moreover, their multisets of poles in $U\cap\{\re s>\ovb{D}-2\}$ coincide.
\end{proposition}

We now introduce the notion of complex dimensions of $(\ty,\O)$ analogously as in the case of ordinary relative fractal drums.

\begin{definition}\label{pcw}
Let $\O\subseteq\eR^N$ be of finite $N$-dimensional Lebesgue measure and such that its Lapidus zeta function at infinity can be meromorphically extended to some open connected neighborhood $W$ of the half-plane $\{\re s\geq\ovb{\dim}_B(\ty,\O)\}$.
We define the {\em set of visible complex dimensions of $(\ty,\O)$ through $W$} as the set of poles of the distance zeta function $\zeta_{\ty,\O}$ that are contained in $W$ and denote it by
\begin{equation}\label{po_vis}
\po(\zeta_{\ty,\O},W):=\{\omega\in W:\mbox{$\omega$ is a pole of }\zeta_{\ty,\O}\}
\end{equation}
which we will abbreviate to $\po(\zeta_{\ty,\O})$ when there is no ambiguity concerning the choice of $W$ (or when $W=\Ce$).

Furthermore, if $\zeta_{\ty,\O}$ possesses a meromorphic continuation to the whole of $\Ce$, we will call the set $\po(\zeta_{\ty,\O},\Ce)$ the {\em set of $($all$)$ complex dimensions of $(\ty,\O)$}.
The subset of $\po(\zeta_{\ty,\O},W)$ consisting of poles with real part equal to $\ovb{\dim}_B(\ty,\O)$ is called the set of {\em principal complex dimensions of $(\ty,\O)$} and is denoted by $\dim_{PC}(\ty,\O)$.
\end{definition}

%
%

\section{Cantor-like Sets at Infinity}\label{qp_sets}

In this section we will construct a subset of $\eR^2$ with prescribed box dimension $D\in(-\ty,-2)$ at infinity that will have a Cantor-like structure in a sense that will be described below.
This set depends on two parameters and is denoted by $\O_{\ty}^{(a,b)}$ in Definition~\ref{Omega(a,b,ty)}.
Furthermore, these sets can be used as building blocks for the construction of (algebraically and transcendentally) quasiperiodic sets at infinity by using some classical results from transcendental number theory (see \cite{ra}).

\begin{figure}[h]
\begin{center}
\includegraphics[width=12cm, height=3.5cm]{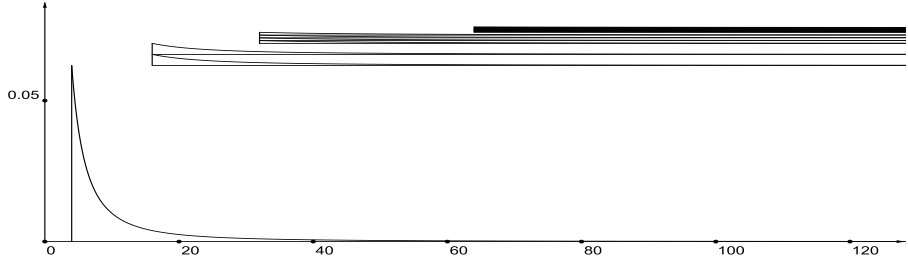}
\end{center}
\caption{An example of the Cantor-like two parameter set $\O_{\ty}^{(a,b)}$ from Definition \ref{Omega(a,b,ty)}. Here, $a=1/4$ and $b=2$. Note that the axes are not in the same scale and only the first four steps in the construction of the set $\O_{\ty}^{(1/4,2)}$ are shown; that is, for $m=1,2,3,4$.}
\label{Cantor_u_besk}
\end{figure}
\begin{definition}\label{Omega(a,b,ty)}
For $a\in(0,1/2)$ and $b\in(1+\log_{1/a}2,+\ty)$ we define a {\em two parameter unbounded set} denoted by $\O_{\ty}^{(a,b)}$.
We start with the countable family of sets
$$
\O_{m}^{(a,b)}:=\{(x,y)\in\eR^2\,:\,x>a^{-m},\ 0<y<x^{-b}\},\quad m\geq 1.
$$
Now, we will construct the set $\O_{\ty}^{(a,b)}$ by ``stacking'' the translated images of the sets $\O_{m}^{(a,b)}$ along the $y$-axis on top of each other.
More precisely, for each $m\geq 1$ we take $2^{m-1}$ copies of $\O_{m}^{(a,b)}$ and arrange all of these sets by vertical translations so that they are pairwise disjoint and lie in the strip $\{0\leq y\leq S\}$.
Here, $S$ is the sum of widths of all of these sets, i.e., 
$
S=\sum_{m=1}^{\ty}2^{m-1}\cdot(a^{-m})^{-b}=\frac{a^{b}}{1-2a^{b}}.
$
Moreover, without loss of generality, we can arrange them in an ``increasing fashion'', i.e., stacking them from bottom to top as $m$ increases $($see Figure~\ref{Cantor_u_besk}$)$.
Finally, we define $\O_{\ty}^{(a,b)}$ as the disjoint union of all of these sets.
\end{definition}

\begin{remark}\label{fin_vol}
The condition $b>1+\log_{1/a}2$ ensures that $\O_{\ty}^{(a,b)}$ has finite Lebesgue measure:
\begin{equation}\label{konacna}\nonumber
\begin{aligned}
|\O_{\ty}^{(a,b)}|&=\sum_{m=1}^{\ty}2^{m-1}|\O_{m}^{(a,b)}|=\sum_{m=1}^{\ty}2^{m-1}\int_{a^{-m}}^{+\ty}x^{-b}\di x=\frac{1}{b-1}\sum_{m=1}^{\ty}2^{m-1}(a^{-m})^{1-b}\\
&=\frac{1}{2(b-1)}\sum_{m=1}^{\ty}(2a^{b-1})^m=\frac{a^{b-1}}{(b-1)(1-2a^{b-1})}.
\end{aligned}
\end{equation}
and the last sum above is convergent for $b>1+\log_{1/a}2$ since then $2a^{b-1}<1$.
\end{remark}

\begin{proposition}\label{twop}
The distance zeta function of the two parameter unbounded set $\O_{\ty}^{(a,b)}$ calculated via the $|\cdot|_{\ty}$-norm on $\eR^2$ is given by
\begin{equation}\label{zeta_ab}
\zeta_{\ty,\O_{\ty}^{(a,b)}}(s;|\cdot|_{\ty})=\frac{1}{s+b+1}\cdot\frac{1}{a^{-(s+b+1)}-2}
\end{equation}
and is meromorphic on $\Ce$. Furthermore, the set of complex dimensions of $\O_{\ty}^{(a,b)}$ at infinity visible through $W:=\{\re s>\log_{1/a}-b-3\}$ is given by
\begin{equation}\label{twop_set}
\{-(b+1)\}\cup\left(\log_{1/a}2-(b+1)+\frac{2\pi}{\log (1/a)}\I\Ze\right).
\end{equation}
Finally, we also have that
$
{{\ovb\dim}_{B}}(\ty,\O_{\ty}^{(a,b)})=\log_{1/a}2-(b+1).
$
\end{proposition}

\begin{proof}
Let us choose $T=1$ and calculate:
$$
\begin{aligned}
\zeta_{\ty,\O_{\ty}^{(a,b)}}(s;1;|\cdot|_{\ty})&=\int_{\O_{\ty}^{(a,b)}}|(x,y)|_{\ty}^{-s-2}\di x\mathrm{d}y=\sum_{m=1}^{\ty}2^{m-1}\int_{\O_{m}^{(a,b)}}x^{-s-2}\di x\mathrm{d}y\\
&=\sum_{m=1}^{\ty}2^{m-1}\int_{a^{-m}}^{+\ty}\di x\int_0^{x^{-b}}x^{-s-2}\di y\\
&=\sum_{m=1}^{\ty}2^{m-1}\int_{a^{-m}}^{+\ty}x^{-s-2-b}\di x=\frac{1}{2(s+b+1)}\sum_{m=1}^{\ty}(2a^{s+b+1})^m\\
&=\frac{1}{s+b+1}\cdot\frac{1}{a^{-(s+b+1)}-2},
\end{aligned}
$$
where the last two equalities follow since $\re s>\log_{1/a}2-(b+1)$.
From this we see that $D(\zeta_{\ty,\O_{\ty}^{(a,b)}}(\,\cdot\,;|\cdot|_{\ty}))=\log_{1/a}2-(b+1)$ and the zeta function has a (unique) meromorphic extension to all of $\Ce$ defined by \eqref{zeta_ab}.
Furthermore, we have that
$
{\ovb{\dim}}_B(\ty,\O_{\ty}^{(a,b)})=\log_{1/a}2-(b+1).
$
Since $({\ty},\O_{\ty}^{(a,b)})$ is contained in a strip of finite width, we can apply Proposition~\ref{euc_ty} to conclude that the difference $\zeta_{\ty,\O_{\ty}^{(a,b)}}(\,\cdot\,;|\cdot|_{\ty})-\zeta_{\ty,\O_{\ty}^{(a,b)}}$ is holomorphic on the half-plane $\{\re s>\log_{1/a}2-(b+1)-2\}=\{\re s>\log_{1/a}2-b-3\}$ from which we conclude that the complex dimensions of $({\ty},\O_{\ty}^{(a,b)})$ visible through $W$ are given by \eqref{twop_set}.
\end{proof}
The two parameter set $\O_{\ty}^{(a,b)}$ is Cantor-like in the sense that its construction parallels, in a way, the construction of the (generalized) Cantor set.
For instance, if we choose $a=3^{-1}$ then the construction of the sets $\O_m^{(1/3,b)}$ for $m\geq 1$ parallels the deletion of the middle-third interval in the standard middle-third Cantor set.
This Cantor-like structure can also be seen in the structure of the complex dimensions of the two sets.
Namely, the set of principal complex dimensions of the middle-third Cantor set is given by $\log_32+\frac{2\pi}{\log 3}\I\Ze$ while the set of principal complex dimensions of $\O_{\ty}^{(1/3,b)}$ is equal to $\log_32-(b+1)+\frac{2\pi}{\log 3}\I\Ze$.
As we can see, the {\em oscillatory period} $\mathbf{p}:=\frac{2\pi}{\log 3}\I\Ze$ of these two sets coincides.
In the definition of fractality proposed in \cite{fzf}, we have defined a set or a relative fractal drum to be fractal if it possesses a nonreal complex dimension.
The motivation for this definition is justified, under mild hypotheses, in the case of relative fractal drums since it is shown in \cite{fzf} that nonreal complex dimensions generate oscillations in the inner geometry of the relative fractal drum.
We expect that analogous results can also be derived in the case of fractal sets at infinity.

It is not difficult to compute the box dimension of $\O_{\ty}^{(a,b)}$ at infinity directly and obtain that ${\dim}_B(\ty,\O_{\ty}^{(a,b)})=\log_{1/a}2-(b+1)$.
Furthermore, one also obtains that $\O_{\ty}^{(a,b)}$ is not Minkowski measurable at infinity which is expected due to the presence of nonreal complex dimensions.
For the detailed calculation see \cite[Example 4.63]{ra}.


\begin{thebibliography}{99}


\bibitem{An} {\sc R.\ Andersson}, {\em Unbounded Soboleff regions}, {Math.\ Scand.} {\bf 13} (1963), 75--89.
 


\bibitem{carlson} {\sc B.C.\ Carlson}, {\em Special Functions of Applied Mathematics}, Academic Press, New York, 1977.


\bibitem{EvJa} {\sc J.-Cl.\ Evard, F.\ Jafari}, {\em A Complex Rolle's Theorem}, {American Mathematical Monthly}, No.\ 9, {\bf 99} (1992), 858--861.




\bibitem{Dol} {\sc C.E.\ Dole}, {\em Flight Theory and Aerodynamics}, John Wiley \& Sons, Inc, New York, 1981.

\bibitem{Dzu} {\sc J.\ D\v zurina}, {\em Unbounded oscillation of the second-order neutral differential equations}, {Mathematica Slovaca} No.\ 4, {\bf 51} (2001), 441--447. 



















%

















\bibitem{Har} {\sc D.A.\ Harville}, {\em Matrix Algebra from a Statistician's Perspective}, Springer-Verlag, New York, 1997.




	



%
%





\bibitem{Hur} {\sc R.\ Hurri},  {\em The weighted Poincar\' e inequalities}, {Ann.\ Acad.\ Sci.\ Fenn.\ Ser.\ A.\ I.\ Math.} {\bf 17} (1992), 409--423.





\bibitem{Karp} {\sc B.\ Karpuz}, {\em Unbounded oscillation of higher-order nonlinear delay dynamic equations of neutral type with oscillating coefficients}, {Electronic Journal of Qualitative Theory of Differential Equations} {\bf 34} (2009), 1--14.


%













%

\bibitem{La} {\sc E.M.\ Landis}, {\em On the behaviour of solutions of higher-order elliptic
equations in unbounded domains}, {Tr.\ Mosk.\ Mat.\ Ob\^ s.} {\bf 31} (1974), 35--58 (in Russian).\ English translation: {Trans.\ Mosc.\ Math.\ Soc.} {\bf 31} (1974), 30--54.




















%












%



%














\bibitem{fzf} {\sc M.L.\ Lapidus, G.\ Radunovi\'c, D.\ \v Zubrini\'c}, {\em Fractal Zeta Functions and Fractal Drums$:$ Higher-Dimensional Theory of Complex Dimensions}, research monograph, Springer, New York, 2016, to appear, approx.\ 620 pages. 














\bibitem{brezish} {\sc M.L.\ Lapidus, G.\ Radunovi\'c, D.\ \v Zubrini\'c}, 
{\em Fractal zeta functions and complex dimensions of relative fractal drums}, {Journal of Fixed Point Theory and Applications} No.\ 2, {\bf 15} (2014), 321--378.\ Festschrift issue in honor of Haim Brezis' 70th birthday. (doi: 10.1007/s11784-014-0207-y.) (Also: e-print, {\tt arXiv:1407.8094v3 [math-ph]}, 2014; IHES preprint, 
{\tt IHES/M/15/14}, 2015.)


\bibitem{tabarz} {\sc M.L.\ Lapidus, G.\ Radunovi\'c, D.\ \v{Z}ubrini\'c}, {\em Fractal zeta functions and complex dimensions: A general higher-dimensional theory}, survey article, in: {\em 
Fractal Geometry and Stochastics V} (C.\ Bandt, K.\ Falconer and M.\ Z\"ahle, eds.), Proc.\ Fifth Internat.\ Conf.\ (Tabarz, Germany, March 2014), {Progress in Probability}, vol.\ {\bf70}, Birkh\"auser/Springer Internat., Basel, Boston and Berlin, 2015, pp.\ 229--257; doi:10.1007/978-3-319-18660-3${}_-$\kern-1pt13.\ (Based on a plenary lecture given by the first author at that conference.) (Also: e-print, {\tt arXiv:1502.00878v3 [math.CV]}, 2015; IHES preprint, 
{\tt IHES/M/15/16}, 2015.)




















%


\bibitem{lapidusfrank12} {\sc M.L.\ Lapidus, M.\ van Frankenhuijsen},
{\em Fractal Geometry, Complex Dimensions
and Zeta Functions$:$
Geometry and Spectra of Fractal Strings}, second revised and enlarged edition (of the 2006 edn.), Springer Monographs in Mathematics, Springer, New York, 2013.






%




\bibitem{Mat} {\sc P.\ Mattila}, {\em Geometry of Sets and Measures in Euclidean
Spaces$:$ Fractals and Rectifiability}, Cambridge Univ. Press, Cambridge, 1995.

\bibitem{Mattn} {\sc L.\ Mattner}, {\em Complex differentiation under the integral}, {Nieuw Arch.\ Wiskd.} (5) No.\ 1, {\bf 2} (2001), 32--35.


\bibitem{May} {\sc J.\ Maynard Smith}, {\em Mathmatical Ideas in Biology}, Cambridge University Press, Cambridge 1968.

\bibitem{Maz1} {\sc V.G.\ Maz'ya}, {\em The Dirichlet problem for elliptic equations of arbitrary
order in unbounded regions}, {Dokl.\ Akad.\ Nauk SSSR} {\bf 150} (1963), 1221--1224 (in Russian).\ English translation: {Sov.\ Math.\ Dokl.} {\bf 4} (1963), 860--863.

\bibitem{mazja}
{\sc V.G.\ Maz'ya}, {\em Sobolev Spaces$:$ with Applications to Elliptic Partial Differential Equations}, second revised and augmented edition, Springer-Verlag, Berlin, 2011.










%














\bibitem{Pou} {\sc S.\ Pourzeynali, T.K.\ Datta}, {\em Control of suspension bridge flutter instability using pole-placement technique}, {Journal of Sound and Vibration} {\bf 282} (2005), 89--109.

\bibitem{Rab} {\sc V.S.\ Rabinovi\v c}, {\em Pseudodifferential equations in unbounded regions with conical structure at infinity}, {Math.\ USSR SB.} No.\ 1, {\bf 9} (1969), 73--92.

\bibitem{ra} {\sc G.\ Radunovi\'c}, {\em Fractal Analysis of Unbounded Sets in Euclidean Spaces and Lapidus Zeta Functions}, Ph.\ D.\ Thesis, University of Zagreb, Croatia, 2015.

\bibitem{razuzup} {\sc G.\ Radunovi\'c, D.\ \v Zubrini\'c, V.\ \v Zupanovi\'c},
{\em Fractal analysis of Hopf bifurcation at infinity}, {Internat.\ J.\ Bifur.\ Chaos Appl.\ Sci.\ Engrg.}\ {\bf 22} (2012), 1230043-1--1230043-15.














\bibitem{She} {\sc G.\ Sheng, L.\ Brown, J.\ Otremba, J.\ Pang, M.S.\ Quatu, R.V.\ Dukkipati}, {\em Chirp, squeal and dynamic instability of misaligned V-ribbed belts in automotive accessory belt drive systems}, {Int. J. of Vehicle Noise and Vibration} No.\ 1, {\bf 3} (2007), 88--105.











%




\bibitem{VoGoLat} {\sc S.K.\ Vodop'yanov, V.M.\ Gol'dshtein, T.G.\ Latfullin},
{\em Criteria for extension of functions of the class $L_2^1$ from unbounded plane domains}, {Sib.\ Mat.\ Zh.} {\bf 20} (1979), 416--419 (in Russian).\ English translation: {Sib.\ Math.\ J.} {\bf 20} (1979), 298--301.




\bibitem{singl} {\sc D.\ \v Zubrini\'c}, {\em Singular sets of Lebesgue integrable functions},
{Chaos, Solitons \&\ Fractals} {\bf 21} (2004), 1281--1287.






\end{thebibliography}
\end{document}